\documentclass[11pt]{article}

\usepackage[mathscr]{eucal}
\usepackage{amsmath,amssymb,amsthm}
\usepackage{graphicx}
\usepackage{subcaption}
\usepackage{color}
\usepackage{multirow}

\usepackage{cite}
\usepackage{hyperref}
\hypersetup{colorlinks, linkcolor=blue, citecolor=blue, urlcolor=blue}

\allowdisplaybreaks

\makeatletter
\long\def\@makecaption#1#2{%
  \vskip\abovecaptionskip\footnotesize
  \sbox\@tempboxa{#1. #2}%
  \ifdim \wd\@tempboxa >\hsize
    #1. #2\par
  \else
    \global \@minipagefalse
    \hb@xt@\hsize{\hfil\box\@tempboxa\hfil}%
  \fi
  \vskip\belowcaptionskip}
\makeatother

\newcommand{\todo}[1][\null]{\ensuremath{\clubsuit}}

\newcommand{\noprint}[1]{}
\newcommand{\checked}[1][\null]{\ensuremath{\boldsymbol{\surd}}}

\newcommand{\red}[1]{{#1}}
\newcommand{\eps}{\varepsilon}
\newcommand{\p}{\partial}
\newcommand{\const}{\mathop{\rm const}\nolimits}

\newtheorem{theorem}{Theorem}[section]
{\theoremstyle{definition}
\newtheorem{example}{Example}[section]
\newtheorem{remark}{Remark}[section]
}

\begin{document}

\title{
Well-balanced mesh-based and meshless schemes for the shallow-water equations
}
\author{ Alexander Bihlo \and Scott MacLachlan}

%\institute{ A. Bihlo \at
%  Department of Mathematics and Statistics\\
%  Memorial University of Newfoundland\\
%  St.\ John's, NL \\
%  A1C 5S7, Canada \\
%  \email{abihlo@mun.ca}
%  \and
%  S. MacLachlan\at
%  Department of Mathematics and Statistics\\
%  Memorial University of Newfoundland\\
%  St.\ John's, NL \\
%  A1C 5S7, Canada \\
%  \email{smaclachlan@mun.ca}  
%}

\maketitle

\begin{abstract}
We formulate a general criterion for the exact preservation of the ``lake at rest'' solution in general mesh-based and meshless numerical schemes for the strong form of the shallow-water equations with bottom topography. The main idea is a careful mimetic design for the spatial derivative operators in the momentum flux equation that is paired with a compatible averaging rule for the water column height arising in the bottom topography source term. We prove consistency of the mimetic difference operators analytically and demonstrate the well-balanced property numerically using finite difference and RBF-FD schemes in the one- and two-dimensional cases.
%\keywords{Meshless finite differences \and Shallow water equations
%  \and Well-balanced schemes}
%\subclass{65M06 \and 76M25 \and 76B15}
\end{abstract}

\section{Introduction}

The shallow-water equations are a central model in geophysical fluid dynamics that is extensively used in the numerical simulation of propagating long waves such as tsunamis. A peculiarity of the shallow-water equations as used in ocean modeling is the presence of a source term arising due to a non-flat ocean bottom topography. An important exact steady-state solution of the shallow-water equations in this context is the \textit{lake at rest solution}, i.e.\ that the total water height over arbitrary bottom topography remains constant and flat in the absence of a horizontal velocity field. From the numerical point of view, exactly preserving the lake at rest solution is usually the first benchmark for the quality of a discretization of the shallow-water equations. Since essentially all small amplitude wave solutions of the shallow-water equations can be regarded as perturbations of the lake at rest solution, the importance of preserving the lake at rest solution exactly cannot be overestimated, as it avoids the occurrence of spurious numerical waves that can render wholly inaccurate computed solutions. Numerical schemes that can preserve the lake at rest solution are called \textit{well-balanced}.
Designing such well-balanced numerical schemes for the shallow-water equations has attracted extensive interest in the literature, in particular in the finite-volume and discontinuous Galerkin methods communities. Examples of well-balanced numerical schemes are reported in, e.g.,~\cite{audu04a,gall03a,kurg07a,leve98a,vate15a}.

Besides finite-volume methods, several nodal-based discretization methodologies approximate derivatives of the field functions at a given point as linear combinations of certain weights with the field functions evaluated at nearby points. Examples for such methods include classical finite differences, meshless finite differences (where the weights are found using polynomial interpolation) and radial basis function based finite differences (RBF-FD; where the weights are found using radial basis function interpolation). It is this framework of nodal-based discretization for the shallow-water equations that we are interested in here.

Several numerical schemes for the shallow-water equations have been constructed over the past 20 years within the framework of the RBF methodology, see e.g.~\cite{hon99a,wong02a,zhou04a}, but to the best of our knowledge none of these papers has explicitly studied the well-balanced properties of the constructed schemes. It was pointed out in~\cite{xia13a} within the framework of a smoothed-particle hydrodynamics scheme that preserving the lake at rest solution in meshless approximations to the shallow-water equations is a nontrivial endeavor. We show in the present paper that constructing well-balanced schemes for the shallow-water equations requires a careful design for the spatial derivatives in the momentum flux equations that is paired with a compatible averaging rule for the water column height arising in the momentum flux source terms. While mostly focusing on finite difference and RBF-FD derivative approximations, the derived conditions are applicable to any nodal-based derivative approximation for the shallow-water equations in the strong form. 

Numerical conservation of physical properties, such as mass, momentum,
and energy, has received significant attention in recent years,
including many contributions to the literature of so-called mimetic
discretizations; see, for example \cite{boch06Ay, hyma99a, brez05a}
and the references therein.  For grid-based schemes, generalizations
of finite-volume (or staggered finite-difference) approaches have been
applied to the shallow-water equations in \cite{vant12a, vanr11a},
yielding schemes that conserve mass, momentum, and energy.  Similar
techniques have also been applied to discretizations in Lagrangian
formulations \cite{cara98, fran03, dubi10}, with conservation of
potential vorticity for the smoothed-particle hydrodynamics
discretization of the shallow-water equations shown in \cite{fran03}.
Very little work has been done, however, in terms of combining the
mimetic methodology with general meshless methods, such as RBF-FD
discretizations.  The work presented here can be seen as a necessary
first step in such development, although much more research is needed
in this direction to see if a similarly broad class of mimetic
meshless schemes can be realized.

The further organization of this paper is as follows. In Section~\ref{sec:WellBalancedShallowWater}, we present the rigorous theoretical analysis underlying well-balanced nodal derivative approximations for the shallow-water equations along with several examples of such well-balanced schemes, both for mesh-based finite differences and meshless RBF-FD approximations. Section~\ref{sec:NumericalSimulationsShallowWater} contains numerical simulations for well-balanced mesh-based and meshless schemes for the one-dimensional shallow-water equations as well as meshless schemes for the two-dimensional shallow-water equations. Section~\ref{sec:ConclusionsShallowWater} is devoted to the conclusions of this work.

\section{Well-balanced shallow-water equation discretizations}\label{sec:WellBalancedShallowWater}

In this section, we review the shallow-water equations together with the lake at rest solution. We then proceed to introduce the general formalism for finding derivative approximations using weighted nodal-based approximations. Within this framework, we then derive general criteria for obtaining well-balanced discretization schemes for the shallow-water equations.

\subsection{The shallow-water equations}

The shallow-water equations with variable bottom topography are given by the transport equations for mass and momentum in the following form~\cite{pedl87Ay},
\begin{equation}\label{eq:ShallowWaterEqs}
 \boldsymbol{\rho}_t+\mathbf{F}_x+\mathbf{G}_y=\mathbf{S},
\end{equation}
where $\rho=(h,hu,hv)^{\rm T}$ is the vector of mass and momentum, $\mathbf{F}=(hu,hu^2+h^2/2,huv)^{\rm T}$ and $\mathbf{G}=(hv,huv,hv^2+h^2/2)^{\rm T}$ are the flux vectors, and $\mathbf{S}=(0,-hb_x,-hb_y)^{\rm T}$ is the source term. Here $h=h(t,x,y)$ denotes the depth of a water column of constant density, $(u,v)^{\rm T}=(u(t,x,y),v(t,x,y))^{\rm T}$ is the (horizontal) vector of vertically averaged fluid velocity and $b=b(x,y)$ is the prescribed bottom topography. Here and in the following, partial derivatives with respect to the independent variables $t$, $x$ and $y$ are denoted by subscripts. Note that for notational simplicity we apply the scaling $g=1$, i.e.\ the gravitational constant is set to one. 

The lake at rest solution is the steady state solution of~\eqref{eq:ShallowWaterEqs} given by
\[
 u=v=0,\quad h+b=\const.
\]
It states that in the absence of horizontal motion, the total height
of the water column and bottom topography over every point in the
spatial domain is constant and independent of time. While it is
usually straightforward to numerically preserve this steady state
solution in the case of flat topography, $b=0$, arbitrary sea bottom elevations are notoriously challenging to handle for typical shallow-water discretization schemes.

\subsection{Computation of weights in nodal-based derivative approximations}\label{subsec:DerivativeApproximations}

The discretization framework we are interested in here is a slight generalization of the framework usually used for conventional finite difference approximations, see e.g.~\cite{forn11b,forn15a,forn15b} for further details and a more in-depth discussion. Suppose we are given $n$ points $x_1<x_2<\cdots <x_n$ covering the (one-dimensional) spatial domain~$\Omega=[x_1,x_n]$ as well as the values of a field function $f(x)$ at these points, $f_j=f(x_j)$, we want to find the weights $w^{\mathcal L}_{ij}$, $i,j=1,\dots,n$, such that for a given linear differential operator~$\mathcal L$ we have
\begin{equation}\label{eq:NodalDerivativeApproximation}
 \mathcal L f|_{x=x_i} \approx \sum_{j=1}^nw^{\mathcal L}_{ij}f_j
\end{equation}
To obtain the weights~$w^{\mathcal L}_{ij}$ for the stencil of the point~$x_i$, one assumes that the approximation~\eqref{eq:NodalDerivativeApproximation} is exact for a given set of basis functions~$\left\{\psi_k(x)\right\}$ over the entire stencil of~$x_i$, i.e. 
\begin{equation}\label{eq:NodalDerivativeApproximationForBasisFunctions}
 \mathcal L \psi_{k}(x_i) = \sum_{j=1}^nw^{\mathcal L}_{ij}\psi_{k}(x_j),\quad k=1,\dots,n.
\end{equation}
This defines a linear system for~$\left\{w^{\mathcal
  L}_{ij}\right\}_{j=1}^n$ for each node $i$.  In general, care must be taken in the choice of nodes and basis functions so that this system has a solution that also yields a stable discretization (see, for example, \cite{seib08}).  Unique solvability is guaranteed when  the coefficient matrix~$(\psi_{k}(x_j))$ (restricted to points where $w^{\mathcal L}_{ij}$ is allowed to be nonzero) is square and non-singular, although more general situations are possible. 

In the following we restrict ourselves to polynomial basis functions,
i.e.\ $\psi_k(x) = x^k$ (when the nodes are on an interval of the real
line), as well as to radial basis functions (RBFs),
$\psi_k(x)=\phi(||x-x_k||)$, although the conditions on well-balanced
shallow-water discretizations derived in
Section~\ref{subsec:WellBalancedConditionsShallowWater} do not depend
on the type of basis functions involved. For RBFs, it is clearly not
essential that $x_i,x_k\in\mathbb{R}$, i.e.\ $x_j$ and $x_k$ could be
vectors $\mathbf{x}_j$ and $\mathbf{x}_k$ in~$\mathbb{R}^d$ as well,
$d\geqslant1$.  For polynomial basis functions, standard
finite-difference discretizations are obtained with appropriately
chosen monomial basis functions on tensor-product meshes, and similar
ideas can be extended to meshless finite-difference schemes assuming
the points are suitably distributed through the domain (see, for
example, \cite{seib08}).
%
%This is not the case for polynomial basis functions since here one typically has to require that higher-dimensional nodal distributions have to constitute a topologically connected, orthogonal grid, which is obtained by taking the tensor product of elementary one-dimensional grids.

For polynomial basis functions in the one-dimensional case, the matrix $(\psi_{k}(x_j))$ in Eq.~\eqref{eq:NodalDerivativeApproximationForBasisFunctions} is the Vandermonde matrix and hence non-singular if the points are distinct. The non-singularity of this matrix is also guaranteed for RBFs, again provided that no two points, $x_i$ and $x_j$, coincide.

It is also possible to consider a family of basis functions that includes both RBFs and polynomials. Such a combination is relevant in meshless RBF schemes as derivative approximations derived solely based on RBFs typically cannot reproduce trivial derivatives such as $\mathcal L c=0$ exactly, for $\mathcal L\in\{\p_x,\p_{xx},\dots\}$ and $c=\const$.

Note that once the system~\eqref{eq:NodalDerivativeApproximationForBasisFunctions} is solved at all points $x_i$, $i=1,\dots,n$, we can assemble the weights~$w^{\mathcal L}_{ij}$ in a differentiation matrix~$W^{\mathcal L}=(w^{\mathcal L}_{ij})$. The derivatives of a field function~$f$ at the nodal points $\mathbf{x}=(x_1,\dots,x_n)^{\rm T}$ are thus approximated as $\mathcal L\mathbf{f}\approx W^{\mathcal L}\mathbf{f}$, where $\mathbf{f}=(f(x_1),\dots,f(x_n))^{\rm T}$.

\subsection{Well-balanced discretizations for the shallow-water equations}~\label{subsec:WellBalancedConditionsShallowWater}

For the sake of simplicity of the following exposition, we consider the one-dimensional form of the shallow-water equation~\eqref{eq:ShallowWaterEqs}, i.e.\
\[
 h_t+(hu)_x=0,\qquad (hu)_t+\left(hu^2+\frac12h^2\right)_x=-hb_x.
\]
Extension to the two-dimensional case is straightforward by enforcing that Eqs.~\eqref{eq:ConditionsOnWellBalancedDerivatives} and~\eqref{eq:ConditionsOnWellBalancedDerivativesMatrix} below have to hold for the $y$-derivative approximation as well.

Since $u=0$ in the lake at rest solution, we need to preserve the property
\[
 \frac12\partial_xh^2=-h\partial_xb,
\]
numerically for the case when $h+b=c$, where $c=\const$.  At the discrete level, this translates to the requirement that, at all \red{nodal points},
\[
 \frac12\left(\mathrm{D}^{\rm f}_xh^2\right)_i=-\overline{h}_i\left(\mathrm{D}^{\rm s}_x(c-h)\right)_i,
\]
is satisfied, where $\mathrm{D}_x^{\rm f}$ and $\mathrm{D}_x^{\rm s}$ are the discrete first derivative operators for the partial derivatives with respect to $x$ arising in the flux and source terms of the shallow-water equations (not necessarily the same), and $\overline{h}_i=\sum_{j=1}^nm_{ij}h_j$ denotes an appropriate average over the field function $h$ in the stencil of $x_i$. Note that consistency of the average requires that $\sum_{j=1}^nm_{ij}=1$.

The above equality is naturally satisfied if the following two conditions hold for all $i$
\begin{equation}\label{eq:ConditionsOnWellBalancedDerivatives}
 \left(\mathrm{D}^{\rm s}_xc\right)_i=0,\qquad \frac12\left(\mathrm{D}^{\rm f}_xh^2\right)_i=\overline{h}_i\left(\mathrm{D}^{\rm s}_xh\right)_i.
\end{equation}
The first condition arises naturally as a consistency condition on $\mathrm{D}^{\rm s}_x$, and it is the second condition that requires more effort to achieve.
A key step to this is to generalize $\mathrm{D}^{\rm f}_xh^2$ so that, rather than discretizing this derivative to act on a vector of values of $h^2$, it acts as a bilinear form,
\[
\left(\frac12\mathrm{D}_x^{\rm f}h^2\right)_i=\frac12\mathbf{h}^{\rm T}W^{\rm f}_i\mathbf{h}.
\]
In this way, we define a {\it differentiation tensor} of order 3, $\mathbf{W}^{\rm f}$, whose $i^{\text{th}}$ slice is the matrix $W^{\rm f}_i$ used above.  There are two important properties of $W^{\rm f}_i$ to note.  First, the classical case, where the derivative operator acts on the vector of values of $h^2$, is still allowed, simply by choosing $W^{\rm f}_i$ to be a diagonal matrix.  Secondly, since we only consider values of $\mathbf{h}^{\rm T}W^{\rm f}_i\mathbf{h}$, only the symmetric part of $W^{\rm f}_i$ matters.  In what follows, we assume $W^{\rm f}_i$ to be symmetric, except where noted.

For the right-hand (source) derivative, $\mathrm{D}_x^{\rm s}$, we use a standard discretization as a matrix, writing
\[
 (\mathrm{D}^{\rm s}_xh)_i=(\mathbf{w}^{\rm s}_i)^{\rm T}\mathbf{h},
\]
where we write $\mathbf{w}^{\rm s}_i=(W_{ij}^{\rm s})_{1\leqslant j\leqslant n}$ for the $i$th row of the matrix~$W^{\rm s}$.  Similarly writing $\mathbf{m}_i = (m_{ij})_{1\leqslant j\leqslant n}$ for the averaging stencil, the second condition in~\eqref{eq:ConditionsOnWellBalancedDerivatives} can be represented as
\[
\frac12\mathbf{h}^{\rm T}W^{\rm f}_i\mathbf{h}=(\mathbf{m}^{\rm T}_i\mathbf{h})((\mathbf{w}^{\rm s}_i)^{\rm T}\mathbf{h}) \text{ for all }i.
\]
From this it follows that
\[
\frac12\mathbf{h}^{\rm T}W^{\rm f}_i\mathbf{h}=\mathbf{h}^{\rm T}(\mathbf{m}^{\rm}_i(\mathbf{w}^{\rm s}_i)^{\rm T})\mathbf{h},
\]
and thus the following relation among the weights in the derivative approximations and the averaging relation has to hold for all $i$:
\begin{equation}\label{eq:ConditionsOnWellBalancedDerivativesMatrix}
W^{\rm f}_i=\mathbf{m}_i\left(\mathbf{w}^{\rm s}_i\right)^{\rm T} + \mathbf{w}^{\rm s}_i\mathbf{m}_i^{\rm T}.
\end{equation}
In other words, specifying an averaging matrix~$M$ and the weights for the derivative matrix, $\mathrm{W}_x^{\rm s}$, Eq.~\eqref{eq:ConditionsOnWellBalancedDerivativesMatrix} prescribes weights in the flux derivative~$\mathrm{D}_x^{\rm f}h^2$, represented by the tensor $\mathbf{W}^{\rm f}$, such that the resulting numerical scheme for the shallow-water equations will be well-balanced.  Alternately, given $\mathbf{W}^{\rm f}$ and one of the matrices $M$ and $W^{\rm s}$, it can be used to check if the other matrix can be defined in such a way as to yield a well-balanced scheme.  We further motivate this approach by stating the following theorem.

\begin{theorem}
\label{thm:wellbalanced}
  Let the derivative tensor, $\mathbf{W}^{\rm f}$, derivative matrix,
  $W^{\rm s}$, and averaging matrix, $M$, be given.  The resulting
  discretization is well-balanced, satisfying
  \eqref{eq:ConditionsOnWellBalancedDerivatives} at every \red{nodal point},
  if and only if $W^{\rm s}\mathbf{1} = \mathbf{0}$ (where
  $\mathbf{1}$ and $\mathbf{0}$ represent the vectors with all entries
  equal to $1$ and $0$, respectively) and, for every $i$,
\begin{subequations}\label{eq:ConditionsOnAveraging}
\begin{equation}
 (\mathbf{w}^{\rm s}_i)^{\rm T}\mathbf{m}_i=\frac{(\mathbf{w}^{\rm s}_i)^{\rm T}W^{\rm f}_i\mathbf{w}^{\rm s}_i}{2(\mathbf{w}^{\rm s}_i)^{\rm T}\mathbf{w}^{\rm s}_i},
\end{equation}
\begin{equation}
 \mathbf{v}^{\rm T}\mathbf{m}_i=\frac{\mathbf{v}^{\rm T}W_i^{\rm f}\mathbf{w}_i^{\rm s}}{(\mathbf{w}_i^{\rm s})^{\rm T}\mathbf{w}^{\rm s}_i}\text{ for any $\mathbf{v}\perp\mathbf{w}^{\rm s}_i$,}
\end{equation}
\end{subequations}
and
\begin{equation}\label{eq:ConsistencyBasis}
\mathbf{v}^{\rm T}W^{\rm f}_i\mathbf{u}=0\text{ for any $\mathbf{u},\mathbf{v}\perp\mathbf{w}^{\rm s}_i$.}
\end{equation}
\end{theorem}

\begin{proof}
First note that $W^s\mathbf{1} = \mathbf{0}$ naturally implies the
first condition in \eqref{eq:ConditionsOnWellBalancedDerivatives}.  Next, from~\eqref{eq:ConditionsOnWellBalancedDerivativesMatrix}, recalling that $W^{\rm f}_i$ is symmetric, we have that
\[
(\mathbf{w}^{\rm s}_i)^{\rm T}W^{\rm f}_i\mathbf{w}^{\rm s}_i=2((\mathbf{w}^{\rm s}_i)^{\rm T}\mathbf{m}_i)\left((\mathbf{w}^{\rm s}_i)^{\rm T}\mathbf{w}_i^{\rm s}\right).
\]
Similarly, we also have that for any $\mathbf{v}\perp\mathbf{w}^{\rm s}_i$,
\[
 \mathbf{v}^{\rm T}W^{\rm f}_i\mathbf{w}_i^{\rm s}=\mathbf{v}^{\rm T}\mathbf{m}_i\left(\mathbf{w}^{\rm s}_i\right)^{\rm T}\mathbf{w}_i^{\rm s}.
\]
Finally, we derive~\eqref{eq:ConsistencyBasis} by noting that~\eqref{eq:ConditionsOnWellBalancedDerivativesMatrix} implies that
\begin{equation}
\mathbf{v}^{\rm T}W^{\rm f}_i\mathbf{u}=0,
\end{equation}
whenever $\mathbf{u},\mathbf{v}\perp\mathbf{w}^{\rm s}_i$.
\end{proof}

When $W^{\rm s}$ and $M$ are specified, Eq.~\eqref{eq:ConditionsOnWellBalancedDerivativesMatrix} directly prescribes the flux differentiation tensor, $\mathbf{W}^{\rm f}$, so that its slices are given by symmetric outer products of the rows of $W^{\rm s}$ and $M$.  When both differentiation rules, $\mathbf{W}^{\rm f}$ and $W^{\rm s}$, are specified, then an algorithmic form of Theorem \ref{thm:wellbalanced} can be expressed by introducing a basis $\langle\mathbf{w}_i^{\rm s},\mathbf{v}_1,\dots,\mathbf{v}_{n-1}\rangle$ where the $n-1$ vectors~$\mathbf{v}_j$ are pairwise orthogonal as well as orthogonal to $\mathbf{w}^{\rm s}_i$, i.e.\
\begin{equation}\label{eq:OrthogonalBasis}
 \mathbf{v}^{\rm T}_j\mathbf{w}^{\rm s}_i=0,\quad \textup{for}\quad 1\leqslant j \leqslant n-1,\qquad \mathbf{v}^{\rm T}_j\mathbf{v}_k=0,\quad \textup{for}\quad j\ne k.
\end{equation}
Then, the existence of a compatible averaging rule is guaranteed by Eq.~\eqref{eq:ConsistencyBasis}, which can be expressed as
\begin{equation}\label{eq:compatibility_req}
\mathbf{v}_j^{\rm T}W^{\rm f}_i\mathbf{v}_k=0,
\end{equation}
for $1\leqslant j,k \leqslant n-1$.  If these conditions are satisfied, then the averaging rule itself is specified for point $i$ by specializing \eqref{eq:ConditionsOnAveraging} to the basis, giving
\begin{equation}\label{eq:defining_m}
 (\mathbf{w}^{\rm s}_i)^{\rm T}\mathbf{m}_i=\frac{(\mathbf{w}^{\rm s}_i)^{\rm T}W^{\rm f}_i\mathbf{w}^{\rm s}_i}{2(\mathbf{w}^{\rm s}_i)^{\rm T}\mathbf{w}^{\rm s}_i}
\text{ and }
 \mathbf{v}_j^{\rm T}\mathbf{m}_i=\frac{\mathbf{v}_j^{\rm T}W_i^{\rm f}(\mathbf{w}_i^{\rm s})}{(\mathbf{w}_i^{\rm s})^{\rm T}\mathbf{w}^{\rm s}_i} \text{ for }1\leqslant j \leqslant n-1.
\end{equation}
The derivation of Eqs.~\eqref{eq:ConditionsOnAveraging}
and~\eqref{eq:ConsistencyBasis}, as well as of the basis discussed
above, can be restated in the obvious way to show both existence and
the definition of a compatible source differentiation matrix, $W^{\rm s}$, given $\mathbf{W}^{\rm f}$ and $M$.

Additionally,
Eq.~\eqref{eq:ConditionsOnWellBalancedDerivativesMatrix} can be
used to understand the consistency and accuracy of the rules derived
as described above in relation to their continuum counterparts.  To do
so, we define vectors $\mathbf{x}^{(p)}$ for $p\geqslant 0$ such that
$\left(\mathbf{x}^{(p)}\right)_i = (x_i)^p$ for all $i$, noting that
the case of $p=0$ corresponds to the vector, $\mathbf{1}$, of all ones.  We define
the following three consistency/accuracy conditions:
\begin{alignat*}{3}
  & A_P:\quad && \mathbf{m}_i^{\rm T}\mathbf{x}^{(p)} = (x_i)^p \quad && 0 \leqslant p
  \leqslant P \\
  & B_P:\quad && (\mathbf{w}^{\rm s}_i)^{\rm T}\mathbf{x}^{(p)} =
  p(x_i)^{p-1} \quad && 0 \leqslant p \leqslant P \\
  & C_{P,Q}:\quad && (\mathbf{x}^{(q)})^{\rm T}W_i^{\rm f}\mathbf{x}^{(p)} =
  (p+q)(x_i)^{p+q-1} \quad && 0 \leqslant p \leqslant P,\ 0 \leqslant q \leqslant Q
\end{alignat*}
For basic consistency, we would require that $A_0$ holds for all $i$
(meaning that $M$ defines a true averaging (row stochastic) matrix),
$B_1$ hold for all $i$ (so that the discrete derivative of a constant
is zero and of a linear function is its slope).  Consistency is
somewhat less natural for $\mathbf{W}^{\rm f}$, and could be expressed
either as $C_{1,0}$ (which is equivalent to $C_{0,1}$ since $W_i^{\rm f}$ is symmetric) or $C_{1,1}$ holding for all $i$.  In the former
case, this requires that the discrete flux derivative reproduce the
true derivative on constants (the case when $p=q=0$) and that the
scheme produced by ``flattening'' the flux derivative matrix at point
$i$ into a row vector as $(\mathbf{x}^{(0)})^{\rm T}W_i^{\rm f}$ exactly
reproduces the derivative of a linear function.  Requiring the stronger
condition $C_{1,1}$ requires the flux derivative also to be faithful to
the true derivative of a quadratic function, which is counter to our
usual expectation of consistency of the first derivative, but more
natural when recalling this is an approximation to the derivative of
$h^2$ and not $h$ itself.  When these conditions hold for all $i$ and
larger values of $P$ and $Q$, they express accuracy conditions that
are natural in the usual sense for meshless finite differences,
requiring that they be accurate pointwise for monomials up to a given
order.

The natural question to be answered in the context of Theorem
\ref{thm:wellbalanced} is whether or not the conditions given there,
together with consistency and accuracy in the sense of conditions
$A_P$, $B_P$, and/or $C_{P,Q}$ yield consistency and accuracy for the
third term in the well-balanced scheme. The following results present
each possible implication.

\begin{theorem}\label{thm:A+B=>C}
  Let $M$ and $W^{\rm s}$ be given, and assume conditions $A_P$ and
  $B_Q$ are satisfied for each \red{nodal point} $i$ with $P,Q\geqslant 0$.  Let $\mathbf{W}^{\rm
    f}$ be determined by
  Eq.~\eqref{eq:ConditionsOnWellBalancedDerivativesMatrix}.  Then
  condition $C_{R,R}$ holds for all $i$ with $R = \min(P,Q)$.
\end{theorem}
\begin{proof}
For any \red{nodal point} $i$, consider $(\mathbf{x}^{(q)})^{\rm T}W_i^{\rm f}\mathbf{x}^{(p)}$ for $0 \leqslant p,q \leqslant \min(P,Q)$:
\begin{align*}
  (\mathbf{x}^{(q)})^{\rm T}W_i^{\rm f}\mathbf{x}^{(p)} & =
  (\mathbf{x}^{(q)})^{\rm T}\mathbf{m}_i(\mathbf{w}^{\rm s}_i)^{\rm T}\mathbf{x}^{(p)}
  + (\mathbf{x}^{(q)})^{\rm T}\mathbf{w}^{\rm s}_i\mathbf{m}_i^{\rm T}\mathbf{x}^{(p)}
  \\
  & = (x_i^q)(px_i^{p-1}) + (qx_i^{q-1})(x_i^p) = (p+q)(x_i)^{p+q-1}.
\end{align*}
\end{proof}

\begin{theorem}\label{thm:A+C=>B}
  Let $M$ and $\mathbf{W}^{\rm f}$ be given, and assume conditions
  $A_R$ and $C_{P,Q}$ are satisfied for each \red{nodal point} $i$ with
  $P,Q,R \geqslant 0$.  Assume there exists a matrix $W^{\rm s}$ such that
  Eq.~\eqref{eq:ConditionsOnWellBalancedDerivativesMatrix} holds
  and the conditions of Theorem \ref{thm:wellbalanced} are
  satisfied. Then condition $B_S$ holds for all $i$ with $S =
  \max(P,Q)$.
\end{theorem}
\begin{proof}
  Without loss of generality, we consider the case where $P\geqslant Q$.
  First consider $(\mathbf{x}^{(q)})^{\rm T}W_i^{\rm f}\mathbf{x}^{(p)}$ for
  $p=q=0$, which gives
\[
  0 = (\mathbf{x}^{(0)})^{\rm T}W_i^{\rm f}\mathbf{x}^{(0)} =
  (\mathbf{x}^{(0)})^{\rm T}\mathbf{m}_i(\mathbf{w}^{\rm s}_i)^{\rm T}\mathbf{x}^{(0)}
  + (\mathbf{x}^{(0)})^{\rm T}\mathbf{w}^{\rm
    s}_i\mathbf{m}_i^{\rm T}\mathbf{x}^{(0)}
  = 2(\mathbf{w}^{\rm s}_i)^{\rm T}\mathbf{x}^{(0)}.
\]
Thus, $(\mathbf{w}^{\rm s}_i)^{\rm T}\mathbf{x}^{(0)} = 0$.  With this, we
consider $(\mathbf{x}^{(q)})^{\rm T}W_i^{\rm f}\mathbf{x}^{(p)}$ for $q=0$,
$1\leqslant p \leqslant P$, giving
\[
px_i^{p-1} = (\mathbf{x}^{(0)})^{\rm T}W_i^{\rm f}\mathbf{x}^{(p)} =
(\mathbf{x}^{(0)})^{\rm T}\mathbf{m}_i(\mathbf{w}^{\rm
  s}_i)^{\rm T}\mathbf{x}^{(p)} + (\mathbf{x}^{(0)})^{\rm T}\mathbf{w}^{\rm
  s}_i\mathbf{m}_i^{\rm T}\mathbf{x}^{(p)} = 1(\mathbf{w}^{\rm
  s}_i)^{\rm T}\mathbf{x}^{(p)} + 0.
\]
Thus, $(\mathbf{w}^{\rm s}_i)^{\rm T}\mathbf{x}^{(p)} = p(x_i)^{p-1}$ for $1
\leqslant p \leqslant P$.
  
\end{proof}

\begin{theorem} \label{thm:B+C=>A}
  Let $W^{\rm s}$ and $\mathbf{W}^{\rm f}$ be given, and assume conditions
  $B_R$ and $C_{P,Q}$ are satisfied for each \red{nodal point} $i$ with
  $P,Q,R \geqslant 1$.  Assume there exists a matrix $M$ such that
  Eq.~\eqref{eq:ConditionsOnWellBalancedDerivativesMatrix} holds
  and the conditions of Theorem \ref{thm:wellbalanced} are
  satisfied. Then condition $A_S$ holds for all $i$ with $S =
  \min(\max(P,Q),R)$.
\end{theorem}
\begin{proof}
  Without loss of generality, we consider the case where $P\geqslant Q$.
  First consider $(\mathbf{x}^{(q)})^{\rm T}W_i^{\rm f}\mathbf{x}^{(p)}$ for
  $p=0$, $q=1$, which gives
\[
  1 = (\mathbf{x}^{(1)})^{\rm T}W_i^{\rm f}\mathbf{x}^{(0)} =
  (\mathbf{x}^{(1)})^{\rm T}\mathbf{m}_i(\mathbf{w}^{\rm s}_i)^{\rm T}\mathbf{x}^{(0)}
  + (\mathbf{x}^{(1)})^{\rm T}\mathbf{w}^{\rm
    s}_i\mathbf{m}_i^{\rm T}\mathbf{x}^{(0)}
  = 0 + 1\mathbf{m}_i^{\rm T}\mathbf{x}^{(0)}.
\]
Thus, $\mathbf{m}_i^{\rm T}\mathbf{x}^{(0)} = 1$.  Similarly, taking
$p=q=1$ gives
\[
  2x_i = (\mathbf{x}^{(1)})^{\rm T}W_i^{\rm f}\mathbf{x}^{(1)} =
  (\mathbf{x}^{(1)})^{\rm T}\mathbf{m}_i(\mathbf{w}^{\rm s}_i)^{\rm T}\mathbf{x}^{(1)}
  + (\mathbf{x}^{(1)})^{\rm T}\mathbf{w}^{\rm
    s}_i\mathbf{m}_i^{\rm T}\mathbf{x}^{(1)}
  = 2\mathbf{m}_i^{\rm T}\mathbf{x}^{(1)}.
  \]
  Thus, $\mathbf{m}_i^{\rm T}\mathbf{x}^{(1)} = x_i$.  Finally, considering
  the general case with $q=1$ and $2\leqslant p \leqslant \min(P,R)$, we have
  \[
(p+1)x_i^{p} = (\mathbf{x}^{(1)})^{\rm T}W_i^{\rm f}\mathbf{x}^{(p)} =
  (\mathbf{x}^{(1)})^{\rm T}\mathbf{m}_i(\mathbf{w}^{\rm s}_i)^{\rm T}\mathbf{x}^{(p)}
  + (\mathbf{x}^{(1)})^{\rm T}\mathbf{w}^{\rm
    s}_i\mathbf{m}_i^{\rm T}\mathbf{x}^{(p)} = px_i^p + \mathbf{m}_i^{\rm T}\mathbf{x}^{(p)}.
  \]
  This implies that $\mathbf{m}_i^{\rm T}\mathbf{x}^{(p)} = x_i^p$ for all $i$.
\end{proof}

These results illustrate a natural asymmetry between the
consistency/accuracy of the various terms defined via
Eq.~\eqref{eq:ConditionsOnWellBalancedDerivativesMatrix} and
Theorem \ref{thm:wellbalanced}.  This is most noticeable in Theorem
\ref{thm:A+C=>B}, where only basic consistency of $M$ is needed for
$W^{\rm s}$ to inherit the full accuracy of $\mathbf{W}^{\rm f}$.  In
contrast, if both $W^{\rm s}$ and $\mathbf{W}^{\rm f}$ are consistent,
Theorem \ref{thm:B+C=>A} shows that $M$ inherits only the lower level
of accuracy from them.

We now provide several examples from classical finite differences on a
uniform mesh with spacing~$\Delta x$ to demonstrate the consequences
of Eq.~\eqref{eq:ConditionsOnWellBalancedDerivativesMatrix} for
prescribing $\mathbf{W}^{\rm f}$ when $W^{\rm s}$ and $M$ are given.
In what follows, $\mathbf{e}_i$ denotes the canonical $i$th unit
vector.

\begin{example}\label{ex:FirstOrderUpwind}
  Suppose we wish to take both derivatives to be given by first-order
  upwind discretizations, with
  \begin{subequations}\label{eq:FirstOrderUpwind}
  \begin{equation}
    \mathbf{w}^{\rm s}_i = \frac{1}{\Delta x}(\mathbf{e}_{i} -
      \mathbf{e}_{i-1})
  \end{equation}
  and
  \begin{equation}
W^{\rm f}_i = \frac{1}{\Delta x}(\mathbf{e}_{i}\mathbf{e}_{i}^{\rm T} -\mathbf{e}_{i-1}\mathbf{e}_{i-1}^{\rm T}).
    \end{equation}
    \end{subequations}
To satisfy the orthogonality condition in \eqref{eq:OrthogonalBasis},
we take $\mathbf{v}_1 = \mathbf{e}_{i-1} + \mathbf{e}_i$, and complete
$\mathbf{v}_2$ through $\mathbf{v}_{n-1}$ with the unit vectors
$\mathbf{e}_{j}$ for $j\neq i$ and $j\neq i-1$. It is straightforward to see that $W^{\rm
  f}_i\mathbf{v}_j = \mathbf{0}$ for $2\leqslant j \leqslant n-1$, meaning that
we only need to verify \eqref{eq:compatibility_req} for $\mathbf{v}_1$
and then use \eqref{eq:defining_m} to define $\mathbf{m}_i =
m_{i-1}\mathbf{e}_{i-1} + m_i\mathbf{e}_i$.

To verify \eqref{eq:compatibility_req}, we see that $W^{\rm
  f}_i\mathbf{v}_1 = (\mathbf{e}_{i} -
\mathbf{e}_{i-1})/\Delta x$ and, so $\mathbf{v}_1^{\rm T}W^{\rm
  f}_i\mathbf{v}_1 = 0$.  Since $W^{\rm f}_i\mathbf{w}^{\rm s}_i =
(\mathbf{e}_{i} + \mathbf{e}_{i-1})/\Delta x^2$, the first
equation in \eqref{eq:defining_m} forces $m_{i-1} = m_i$.  Computing from
the second, we find that these both take value $1/2$, giving
$\mathbf{m}_i = (\mathbf{e}_i + \mathbf{e}_{i-1})/2$.  Direct
calculation shows that
Eq.~\eqref{eq:ConditionsOnWellBalancedDerivativesMatrix} is
satisfied for this $\mathbf{m}_i$.

Written in component form, the associated upwind scheme defined through~\eqref{eq:FirstOrderUpwind} reads
\[
 \frac12\left(\frac{h_i^2-h_{i-1}^2}{\Delta x}\right)=\bar h_i\frac{b_i-b_{i-1}}{\Delta x},\quad \bar h_i=\frac12(h_i+h_{i-1}).
\]
which is directly seen to be well-balanced.  From
\eqref{eq:FirstOrderUpwind}, we can verify that conditions $B_1$ and
$C_{1,0}$ are satisfied for all $i$ by the upwind discretizations.  In this case
(since $C_{1,1}$ does not hold for all $i$), Theorem \ref{thm:B+C=>A}
does not apply, and it can easily be seen that $M$ satisfies $A_0$ for
all $i$, but not $A_1$.  Considering the alternate implications, if we
were to specify $M$ and $W^{\rm s}$, Theorem \ref{thm:A+B=>C} would
confirm that $A_0$ and $B_1$ for all $i$ implies $C_{0,0}$ for all
$i$, but not $C_{1,1}$ for all $i$.  Similarly, since $A_0$ and
$C_{1,0}$ hold for all $i$, Theorem \ref{thm:A+C=>B} implies that
$B_1$ holds for all $i$.
\end{example}

\begin{example}\label{ex:CenteredFiniteDifferences}
A similar calculation verifies that taking a centered averaging for
$\mathbf{m}_i$ yields a well-balanced scheme when the two derivatives
are approximated by centered finite differences.  Setting
\begin{subequations}\label{eq:WeightedCenteredDifferences}
\begin{equation}
 \mathbf{m}_i=\frac12(\mathbf{e}_{i+1}+\mathbf{e}_{i-1}),\quad \mathbf{w}^{\rm s}_i=\frac1{2\Delta x}(\mathbf{e}_{i+1}-\mathbf{e}_{i-1}),
\end{equation}
 then Eq.~\eqref{eq:ConditionsOnWellBalancedDerivativesMatrix}, gives
\begin{equation}
  W^{\rm f}_i=\frac1{2\Delta x}(\mathbf{e}_{i+1}\mathbf{e}_{i+1}^{\rm T}-\mathbf{e}_{i-1}\mathbf{e}_{i-1}^{\rm T} ).
\end{equation}
\end{subequations}
Component-wise the differentiation and averaging
rule~\eqref{eq:WeightedCenteredDifferences} imply that, at the node
$x_i$, our well-balanced scheme is given by
\[
 \frac12\left(\frac{h_{i+1}^2-h_{i-1}^2}{2\Delta x}\right)=\bar h_i\frac{b_{i+1}-b_{i-1}}{2\Delta x},\quad \bar h_i=\frac12(h_{i+1}+h_{i-1}).
\]
which obviously satisfies the
conditions~\eqref{eq:ConditionsOnWellBalancedDerivatives}.
Considering the consistency/accuracy conditions for these rules, we
can directly verify that $A_1$, $B_2$, $C_{1,1}$, and $C_{2,0}$ hold
for all $i$.  (Note that neither $C_{1,1}$ nor $C_{2,0}$ implies the
other, and that $C_{2,1}$ does not hold for all $i$ for this choice of
$\mathbf{W}^{\rm f}$.)  Theorem \ref{thm:A+B=>C} states that $A_1$ and
$B_2$ together imply $C_{1,1}$, Theorem \ref{thm:A+C=>B} states that
$A_1$ and $C_{2,0}$ together imply $B_2$ , and Theorem
\ref{thm:B+C=>A} states that $B_2$ and $C_{1,1}$ imply $A_1$.  We note
that the conclusions of these theorems naturally depend differently on
$P$ and $Q$ in $C_{P,Q}$, with $\max(P,Q)$ appearing in Theorem
\ref{thm:A+C=>B}, but $\min(P,Q)$ in Theorem \ref{thm:B+C=>A}.
\end{example}

\begin{example}
When we choose centered differencing for the source derivative,
\begin{equation}
\mathbf{w}^{\rm s}_i=\frac1{2\Delta x}(\mathbf{e}_{i+1}-\mathbf{e}_{i-1}),
\end{equation}
we can consider which values for $W^{\rm f}_i$ are possible to achieve a well-balanced scheme. If we restrict $W^{\rm f}_i$ to have a
nonzero pattern over only the three points $i-1$, $i$, and $i+1$, we
can write
\begin{align}\label{eq:general_left}
W^{\rm f}_i & = w_{i-1,i-1}\mathbf{e}_{i-1}\mathbf{e}_{i-1}^{\rm T} +
w_{i,i}\mathbf{e}_{i}\mathbf{e}_{i}^{\rm T} +
w_{i+1,i+1}\mathbf{e}_{i+1}\mathbf{e}_{i+1}^{\rm T} \\
& + w_{i-1,i}(\mathbf{e}_{i-1}\mathbf{e}_i^{\rm T} +
\mathbf{e}_{i}\mathbf{e}_{i-1}^{\rm T}) + 
w_{i-1,i+1}(\mathbf{e}_{i-1}\mathbf{e}_{i+1}^{\rm T} +
\mathbf{e}_{i+1}\mathbf{e}_{i-1}^{\rm T}) + 
w_{i,i+1}(\mathbf{e}_{i}\mathbf{e}_{i+1}^{\rm T} +
\mathbf{e}_{i+1}\mathbf{e}_{i}^{\rm T}). \nonumber
\end{align}
Take $\mathbf{v}_1 = \mathbf{e}_i$, $\mathbf{v}_2 = \mathbf{e}_{i-1} +
\mathbf{e}_{i+1}$, and complete $\mathbf{v}_3$ through
$\mathbf{v}_{n-1}$ with $\mathbf{e}_j$ for $j\neq i$ and $j\neq i\pm 1$.  From
\eqref{eq:compatibility_req}, we have $\mathbf{v}^{\rm T}_1W^{\rm
  f}_i\mathbf{v}_1 = w_{i,i} = 0$, $\mathbf{v}^{\rm T}_2W^{\rm
  f}_i\mathbf{v}_1 = w_{i-1,i} + w_{i,i+1} = 0$, and $\mathbf{v}^{\rm
  T}_2W^{\rm f}_i\mathbf{v}_2 = w_{i-1,i-1} + 2w_{i-1,i+1} +
w_{i+1,i+1} = 0$.  Simplifying \eqref{eq:general_left}, we then see a
restricted form of
\begin{align*}
W^{\rm f}_i = w_{i-1,i-1}\mathbf{e}_{i-1}\mathbf{e}_{i-1}^{\rm T} & +
w_{i+1,i+1}\mathbf{e}_{i+1}\mathbf{e}_{i+1}^{\rm T} + w_{i-1,i}(\mathbf{e}_{i-1}\mathbf{e}_i^{\rm T} +
\mathbf{e}_{i}\mathbf{e}_{i-1}^{\rm T}- \mathbf{e}_{i}\mathbf{e}_{i+1}^{\rm T} -
\mathbf{e}_{i+1}\mathbf{e}_{i}^{\rm T})\\
& -\frac{1}{2}(w_{i-1,i-1}+w_{i+1,i+1})(\mathbf{e}_{i-1}\mathbf{e}_{i+1}^{\rm T} +
\mathbf{e}_{i+1}\mathbf{e}_{i-1}^{\rm T}). \nonumber
\end{align*}
In order to not break the flux form of the shallow-water equations, we
need the off-diagonal terms in $W^{\rm f}_i$ to vanish, forcing both
$w_{i-1,i+1} = 0$ and $w_{i-1,i-1} = -w_{i+1,i+1}$. In other words,
the only consistent well-balanced discretization in flux form that
uses centered differences for the source derivatives occurs when also
using centered differences for the flux derivative, as in Example
\ref{ex:CenteredFiniteDifferences}.  Even if we were to allow breaking
of the flux form, straightforward calculation shows that we cannot
enforce consistency condition $C_{1,1}$ without also requiring that
$w_{i-1,i+1} = 0$ and $w_{i-1,i-1} = -w_{i+1,i+1}$.
\end{example}

\begin{example}
To extend the above example, we consider the case where the right-hand (source) derivative is given by centered differencing, but the left-hand (flux) derivative is given by second-order upwinding, with
\begin{equation}
W^{\rm f}_i = \frac{1}{2\Delta x}(3\mathbf{e}_{i}\mathbf{e}_{i}^{\rm
  T} -4\mathbf{e}_{i-1}\mathbf{e}_{i-1}^{\rm T}  + \mathbf{e}_{i-2}\mathbf{e}_{i-2}^{\rm T} ).
\end{equation}
Note that $\mathbf{e}_i^{\rm T}\mathbf{w}^{\rm s}_i = 0$, but
$\mathbf{e}_i^{\rm T}W^{\rm f}_i\mathbf{e}_i = 3/(2\Delta x)$.
Thus, Theorem \ref{thm:wellbalanced} states that no possible choice of
$\mathbf{m}_i$ exists that yields a well-balanced scheme.
\end{example}

\begin{example}
We now consider the case where both derivatives are given by
second-order upwinding, with 
\begin{equation}
W^{\rm f}_i = \frac{1}{2\Delta x}(3\mathbf{e}_{i}\mathbf{e}_{i}^{\rm
  T} -4\mathbf{e}_{i-1}\mathbf{e}_{i-1}^{\rm T}  + \mathbf{e}_{i-2}\mathbf{e}_{i-2}^{\rm T} ),
\end{equation}
and
\begin{equation}
\mathbf{w}^{\rm s}_i = \frac{1}{2\Delta x}(3\mathbf{e}_{i}
-4\mathbf{e}_{i-1}  + \mathbf{e}_{i-2}).
\end{equation}
Note that we can naturally take $\mathbf{v}_1 = \mathbf{e}_i +
\mathbf{e}_{i-1} + \mathbf{e}_{i+1}$, and can find $\mathbf{v}_2$
orthogonal to both $\mathbf{w}^{\rm s}_i$ and $\mathbf{v}_1$ by taking
the cross-product of the two three-dimensional restrictions of these
vectors, giving $\mathbf{v}_2 = 5\mathbf{e}_i + 2\mathbf{e}_{i-1} -
7\mathbf{e}_{i-2}$.  By construction, $\mathbf{v}_2^{\rm
  T}\mathbf{w}^{\rm s}_i = 0$, but $\mathbf{v}_2^{\rm T}W^{\rm
  f}_i\mathbf{v}_2 = 108/(2\Delta x) \neq 0$.  Thus, Theorem
\ref{thm:wellbalanced} again states that no possible choice of
$\mathbf{m}_i$ exists that yields a well-balanced scheme for these
choices.
\end{example}

\red{
The last two examples raise the question of whether higher-order
well-balanced schemes are possible.  This is easily addressed as a
consequence of Equation
\eqref{eq:ConditionsOnWellBalancedDerivativesMatrix}.}

\red{
\begin{theorem}
Let the differentiation tensor, $\mathbf{W}^{\rm f}$, be given.  If
there exists an $i$ such that $W_i^{\rm f}$ is a diagonal matrix with
more than 2 nonzero entries, then no well-balanced scheme exists.
\end{theorem}
\begin{proof}
Equation \eqref{eq:ConditionsOnWellBalancedDerivativesMatrix} states
that a scheme is well-balanced if and only if the symmetric part of
$W_i^{\rm f}$ is a rank-two matrix for all $i$.  When $W_i^{\rm f}$ is
diagonal, then it is its own symmetric part.  The rank of a diagonal
matrix equals the number of nonzero entries in the matrix.  Thus, if
more than two nonzero entries appear in such a $W_i^{\rm f}$, the
scheme cannot be well-balanced.
\end{proof}
}

\red{
This result highlights another asymmetry in the construction of
well-balanced schemes, that the flux derivative cannot be freely
prescribed.  In particular, for flux form discretizations, where
$W_i^{\rm f}$ is constrained to be diagonal, no higher-order
finite-difference stencil can be accommodated under the restriction of
only two nonzero weights.  In constrast, Equation
\eqref{eq:ConditionsOnWellBalancedDerivativesMatrix} and Theorem
\ref{thm:A+B=>C} state that for {\it any} choice of source derivative,
$W^{\rm s}$, and averaging matrix, $M$, a well-balanced scheme can be
defined, inheriting the lower of the consistency orders of $W^{\rm s}$
and $M$, albeit with no expectation that $W_i^{\rm f}$ be in flux
form.  Thus, of the possible ways to complete a well-balanced scheme,
we exclusively adopt this latter one, which allows us to make free
choices of $W^{\rm s}$ and $M$, prescribing a well-balanced scheme
via Equation \eqref{eq:ConditionsOnWellBalancedDerivativesMatrix}.
}

\begin{remark}\label{rem:OnMonomial}
It follows from the first condition in~\eqref{eq:ConditionsOnWellBalancedDerivatives} that well-balanced shallow-water equation discretizations need to employ derivative approximations that are exact for constants. Thus, if the RBF-FD or global RBF collocation methodology is invoked, the underlying RBF interpolant for the field function~$f$ should be of the form $f(x)=\sum_{i=1}^n\alpha_i\phi(||x-x_i||)+\alpha_{n+1}$ with the constraint that $\sum_{i=1}^n\alpha_i=0$. In other words, the RBF basis should be supplemented with the monomial~$\{1\}$. \red{Higher-order polynomials can be included in the basis as well for accuracy considerations, see, e.g.,~\cite{bayo17a}. Such higher-order polynomials may play an important role for the accurate representation of more complicated, non-stationary solutions of the shallow-water equations.}
\end{remark}

\begin{remark}
 It is well-known that the application of RBF-FD methods to purely
 convective PDEs is prone to numerical instabilities since the
 eigenvalues of the differentiation matrices tend to scatter to the
 right half of the complex plane. As a remedy, the inclusion of
 hyperviscosity was proposed~\cite{forn11a}, which allows shifting the
 eigenspectrum of the convective operators back into the left-half
 plane, thus allowing for the use of explicit time-stepping
 methods. We note here that adding hyperviscosity in the momentum
 equations (specifically, terms like $\Delta^k(uh)$, where $\Delta^k$
 is the $k$-th power of the Laplacian operator) to the above
 well-balanced schemes is perfectly possible without tampering with
 the well-balanced property for the lake at rest solution (where $u=0$).
\end{remark}

\begin{remark}
  \red{ As mentioned above, we note that the extension of these
    results to the two-dimensional case follows simply by applying the
    one-dimensional results twice, independently in each coordinate
    direction.  This follows from substituting the lake-at-rest
    solution into Equation \eqref{eq:ShallowWaterEqs}, yielding two
    independent conditions, that $\frac{1}{2}\frac{\partial
      h^2}{\partial x} = h\frac{\partial h}{\partial x}$ and
    $\frac{1}{2}\frac{\partial h^2}{\partial y} = h\frac{\partial
      h}{\partial y}$.  Thus, no cross-derivative terms or other
    coupling arises in the development of well-balanced schemes in two
    dimensions.}
\end{remark}

\section{Numerical simulations}\label{sec:NumericalSimulationsShallowWater}

In this section, we present some numerical verification for the above theoretical construction of well-balanced schemes for the shallow-water equations in the one- and two-dimensional case. 

\subsection{One-dimensional lake at rest solution}
\label{ssec:1Dlake}

We solve the shallow-water equations using either centered finite
differences with the averaging rule as defined in
Example~\ref{ex:CenteredFiniteDifferences}, or with the RBF-FD
method. \red{In the latter case, we exclusively use the multiquadric RBF, i.e.\ $\phi(r)=\sqrt{1+(\epsilon r)^2}$, augmented with the monomial~$\{1\}$ (see Remark~\ref{rem:OnMonomial}), where we set the shape
parameter~$\epsilon=0.1$ in all experiments.} The stencil size of the
RBF-FD method is three (center point and the immediate neighbors to
the left and to the right) and the averaging rule is a normalized
Gaussian filter, \red{with weights given by $C_ie^{-|x_j-x_i|}$ assigned at
all points, $j$, that appear in the stencil for point $i$, and
constant $C_i$ chosen so these weights sum to 1}. The discretization for the flux
derivative~$\tfrac12\mathrm{D}_x h^2$ is then computed using the
condition~\eqref{eq:ConditionsOnWellBalancedDerivativesMatrix}.  \red{We note the
  flexibility in the framework defined above allows us to
  independently choose the source derivative and averaging rule, and
  this approach will always yield a well-balanced scheme.} As
time stepping, we use the Heun scheme. A total of $n=100$ \red{equally
  spaced} points are
used on the domain $\Omega=[-3,3]$ where reflective boundary
conditions were employed. The bottom topography is a cosine bump of
amplitude $A=7$ extending over the interval $[-1,1]$, which is
superimposed with white noise generated independently at each node using normally distributed random numbers with zero mean and unit variance. The initial total water height is $h_0+b=10$.

In Figure~\ref{fig:RBF1D} we show the results of the numerical
computations at $t=10$ using the RBF-FD method. The results of the
classical centered finite difference scheme presented in Example~\ref{ex:CenteredFiniteDifferences} are essentially the same,
with slightly smaller errors overall, and are, hence, not displayed
here. 

\red{Note that for irregular nodal layouts, the eigenvalues of the derivative matrices for the RBF-FD method scatter into the right
half of the complex plane.  This is particularly prominent in the
multidimensional case and in the case that several neighboring nodal
points are used for the computation of the derivative matrices. To
improve the stability of the numerical schemes in these cases,
hyperviscosity or other stabilization should be used.}

\begin{figure}[!ht]
    \centering
    \begin{subfigure}[b]{0.45\textwidth}
            \centering
            \includegraphics[width=\textwidth]{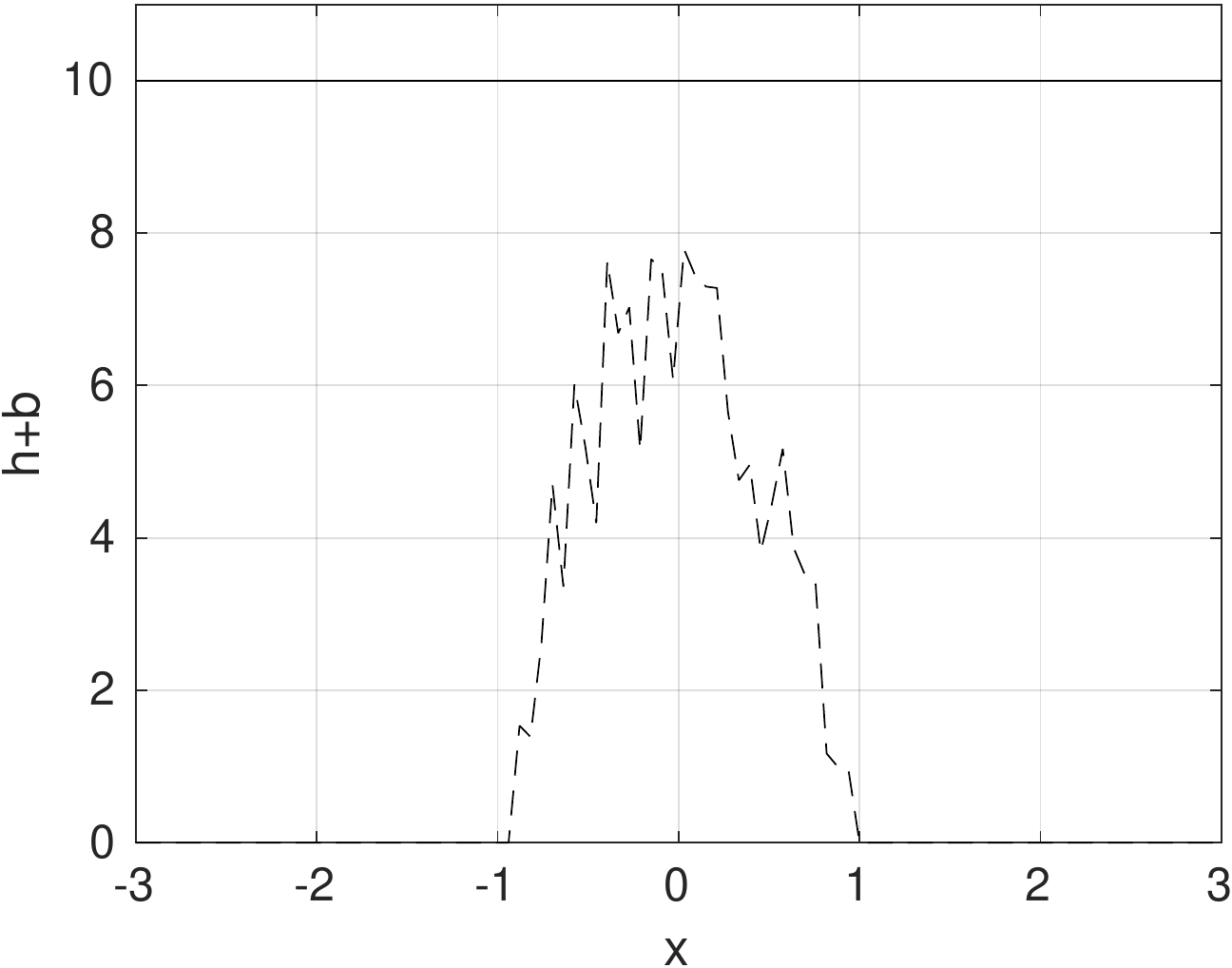}
    \end{subfigure}\hspace{1cm}
\begin{subfigure}[b]{0.45\textwidth}
            \centering
            \includegraphics[width=\textwidth]{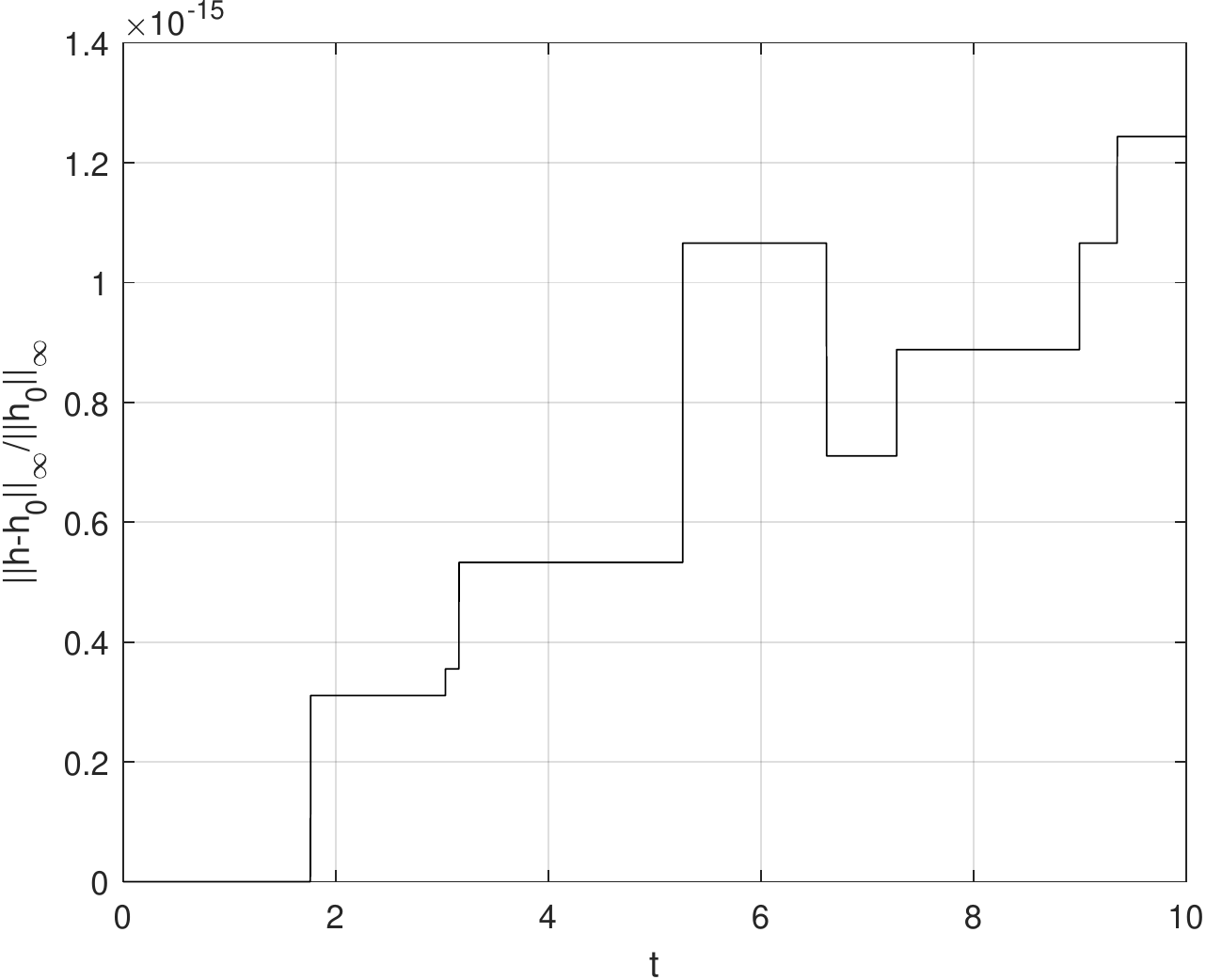}
    \end{subfigure}
    \caption{Numerical integration of the shallow-water equations
      using the \textit{balanced} RBF-FD method on $n=100$ regularly
      spaced nodes, integrated up to $t=10$ with the Heun
      scheme. \textbf{Left:} Total water height at $t=10$ (solid line)
      and bottom topography (dashed line). \textbf{Right:} Relative $l_\infty$-error in the water height.}
    \label{fig:RBF1D}
\end{figure}

Figure~\ref{fig:RBF1D} shows that the RBF-FD scheme is indeed well-balanced, being able to maintain the constant water height even in the presence of quite rough bottom topography. We also monitored the conservation of total mass $\mathcal M=\sum_{i}h_i\Delta x_i$ and found conservation with relative errors of the magnitude $10^{-16}$ and hence machine precision (not shown here). The conservation of mass for the lake at rest solution is a particularly nice feature of the present well-balanced schemes as it is straightforward to check that mass is in general not conserved in numerical schemes for the shallow-water equations using the RBF-FD methodology.

In contrast, Figure~\ref{fig:RBF1Dbad} depicts the results obtained
when the standard RBF-FD approximation is used for both the source
and flux derivatives (i.e., not employing the well-balanced condition
derived above). As expected, the violation of balance leads to the
emergence of spurious waves that travel through the entire
computational domain. The emergence of these waves is typically not
tolerable in numerical schemes for the shallow-water equations, as
they can lead to wrong run-up heights and, thus, to unphysical results
estimating factors such as tsunami inundation or the stress on coastal structures. 

\begin{figure}[!ht]
    \centering
    \begin{subfigure}[b]{0.45\textwidth}
            \centering
            \includegraphics[width=\textwidth]{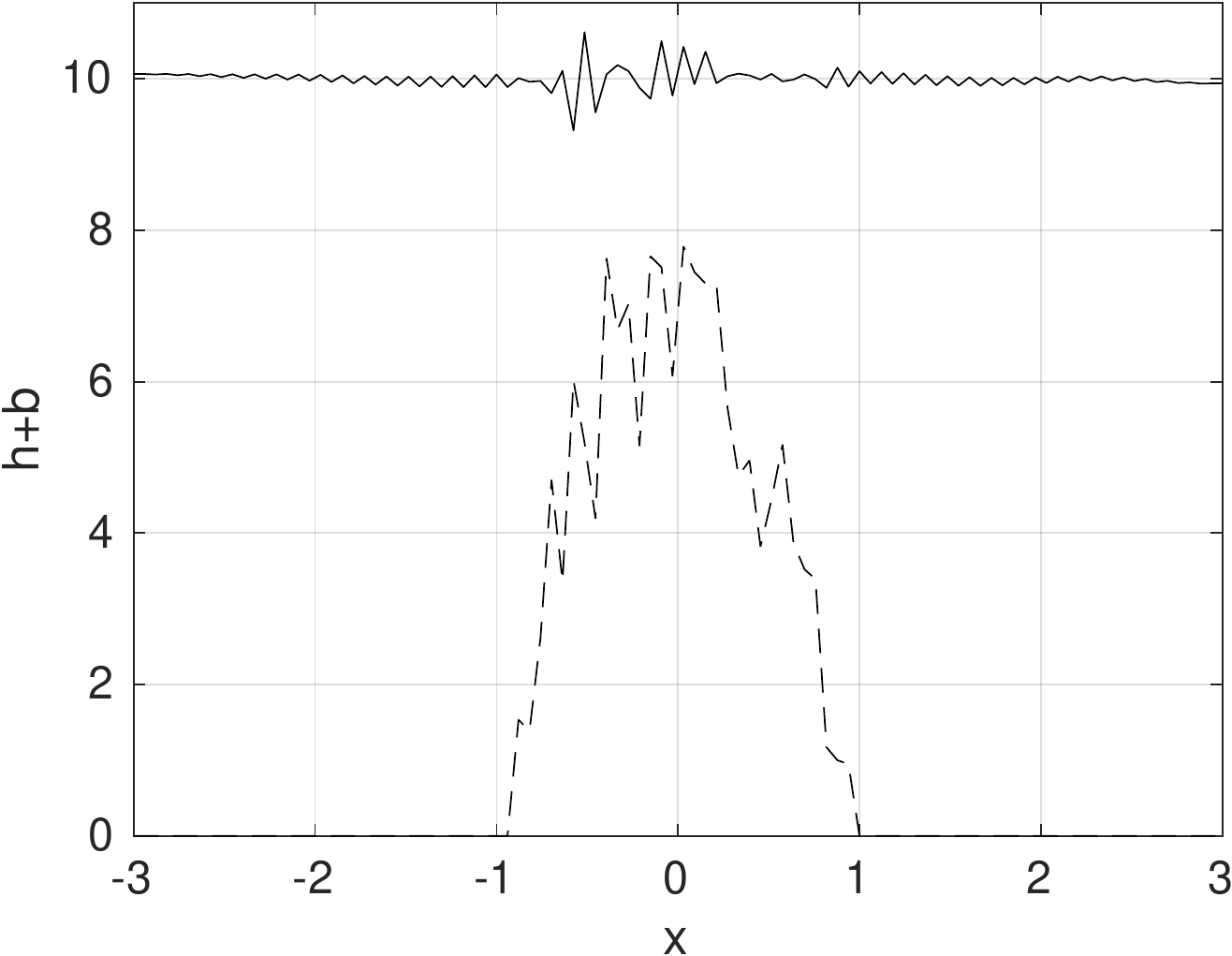}
    \end{subfigure}\hspace{1cm}
\begin{subfigure}[b]{0.45\textwidth}
            \centering
            \includegraphics[width=\textwidth]{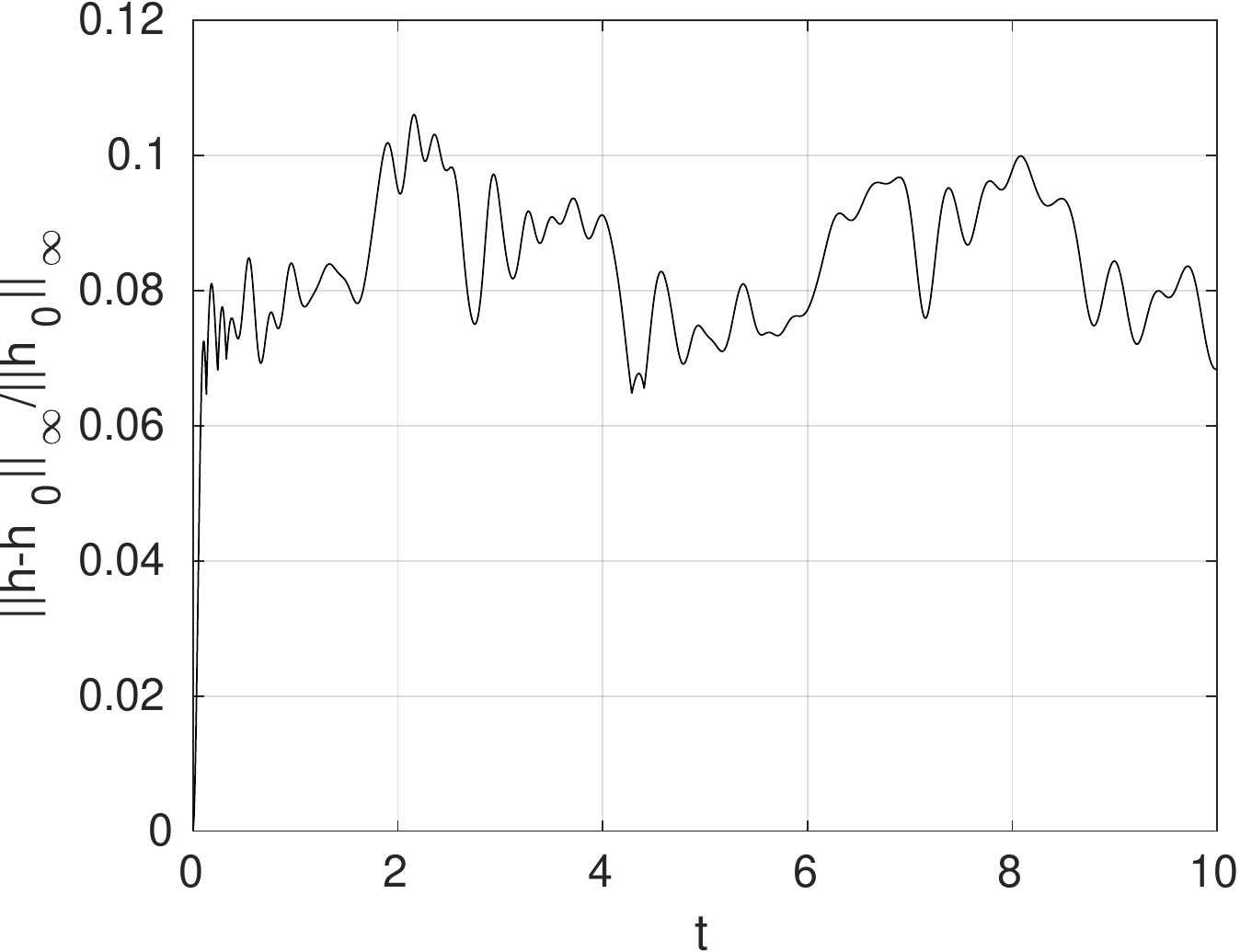}
    \end{subfigure}
    \caption{Numerical integration of the shallow-water equations
      using the \textit{unbalanced} RBF-FD method on $n=100$ regularly
      spaced nodes, integrated up to $t=10$ with the Heun
      scheme. \textbf{Left:} Total water height at $t=10$ (solid line)
      and bottom topography (dashed line). \textbf{Right:} Relative $l_\infty$-error in the water height.}
    \label{fig:RBF1Dbad}
\end{figure}

\subsection{Two-dimensional lake at rest solution}

In the two-dimensional case, we constrain ourselves to the use of the RBF-FD method only. We consider the domain $\Omega=[-3,3]\times[-3,3]$ covered by $n=1600$ nodes. To demonstrate the versatility and independence of the chosen nodal layout (mesh-based or meshfree) of the condition~\eqref{eq:ConditionsOnWellBalancedDerivativesMatrix}, we add $(0.1\Delta x \mathcal N(0,1),0.1\Delta y \mathcal N(0,1))$ as disturbance to each \red{nodal point} originally lying on an orthogonal and equally spaced mesh, where $\mathcal N(0,1)$ is a normally distributed random variable with zero mean and variance one. The resulting nodal layout is depicted in Figure~\ref{fig:NodalLayout}.

\begin{figure}[!ht]
\centering
\includegraphics[scale=0.5]{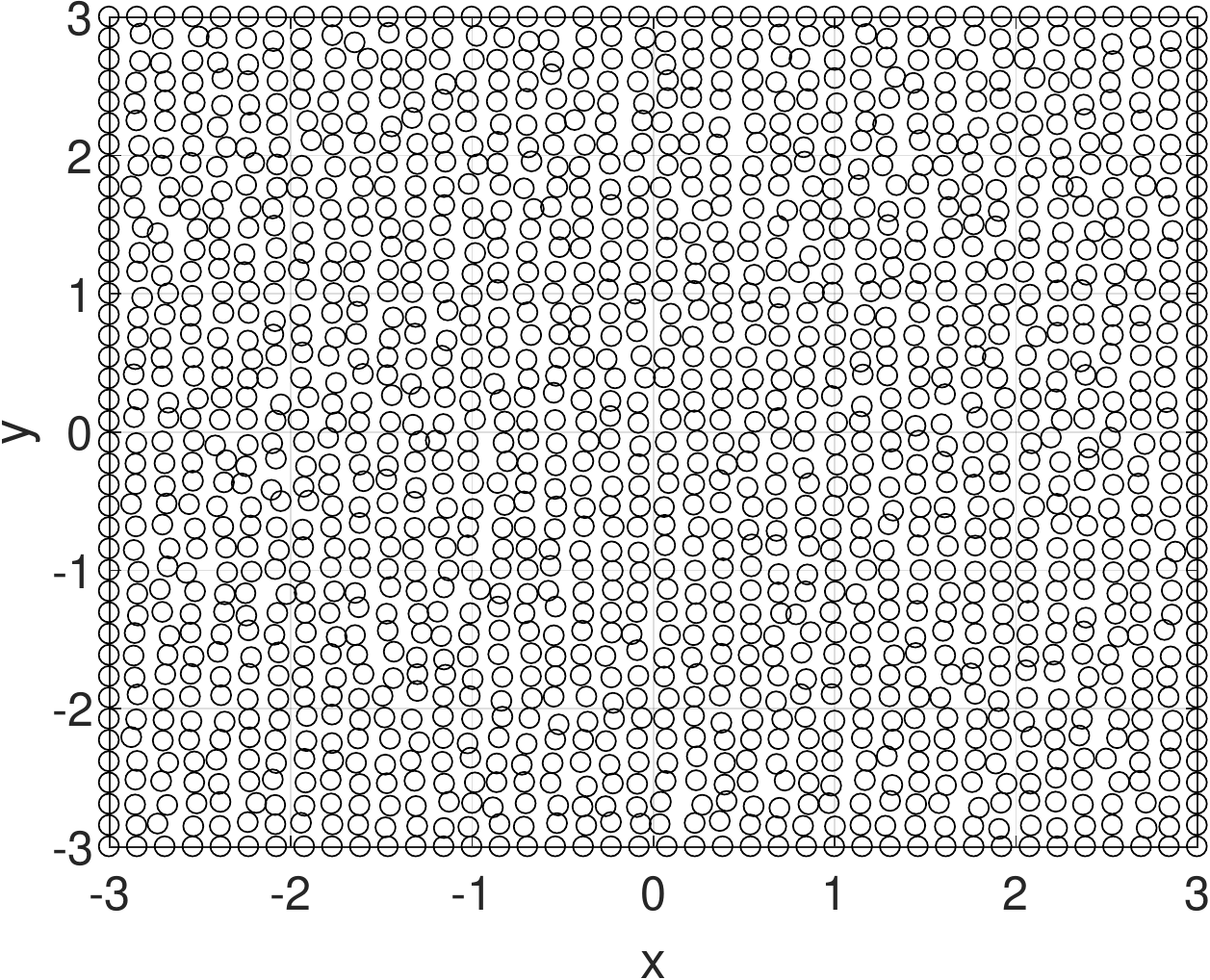}
\caption{Nodal distribution for the balanced 2D shallow-water equations discretization.}
\label{fig:NodalLayout}
\end{figure}

The bottom topography used is a cosine bell on the area
$[-1,1]\times[-1,1]$ with amplitude $A=7$ and again superimposed with
white noise obtained independently at each node from normally distributed random numbers with zero mean and unit variance.

We choose a stencil based on the 25 nearest neighbors of each point in the domain~$\Omega$ for the computation of the RBF-FD differentiation matrices. Once again, the multiquadric RBF is used in all computations and while the current nodal layout might profit from a spatially variable shape parameter~$\epsilon$ for improved accuracy, \red{for the sake of simplicity we used $\eps=1$ in all points.} The averaging rule is a two-dimensional Gaussian filter over the 25 points in the stencil of each point. Once again, the flux derivative discretizations for~$\tfrac12\mathrm{D}_xh^2$ and $\tfrac12\mathrm{D}_yh^2$ are obtained from the condition~\eqref{eq:ConditionsOnWellBalancedDerivativesMatrix}. We integrate the two-dimensional shallow-water equations with the Heun scheme up to $t=10$. The results of this integration are displayed in Figure~\ref{fig:RBF2D} and verify numerically that the scheme is indeed well-balanced.

\begin{figure}[!ht]
    \centering
    \begin{subfigure}[b]{0.45\textwidth}
            \centering
            \includegraphics[width=\textwidth]{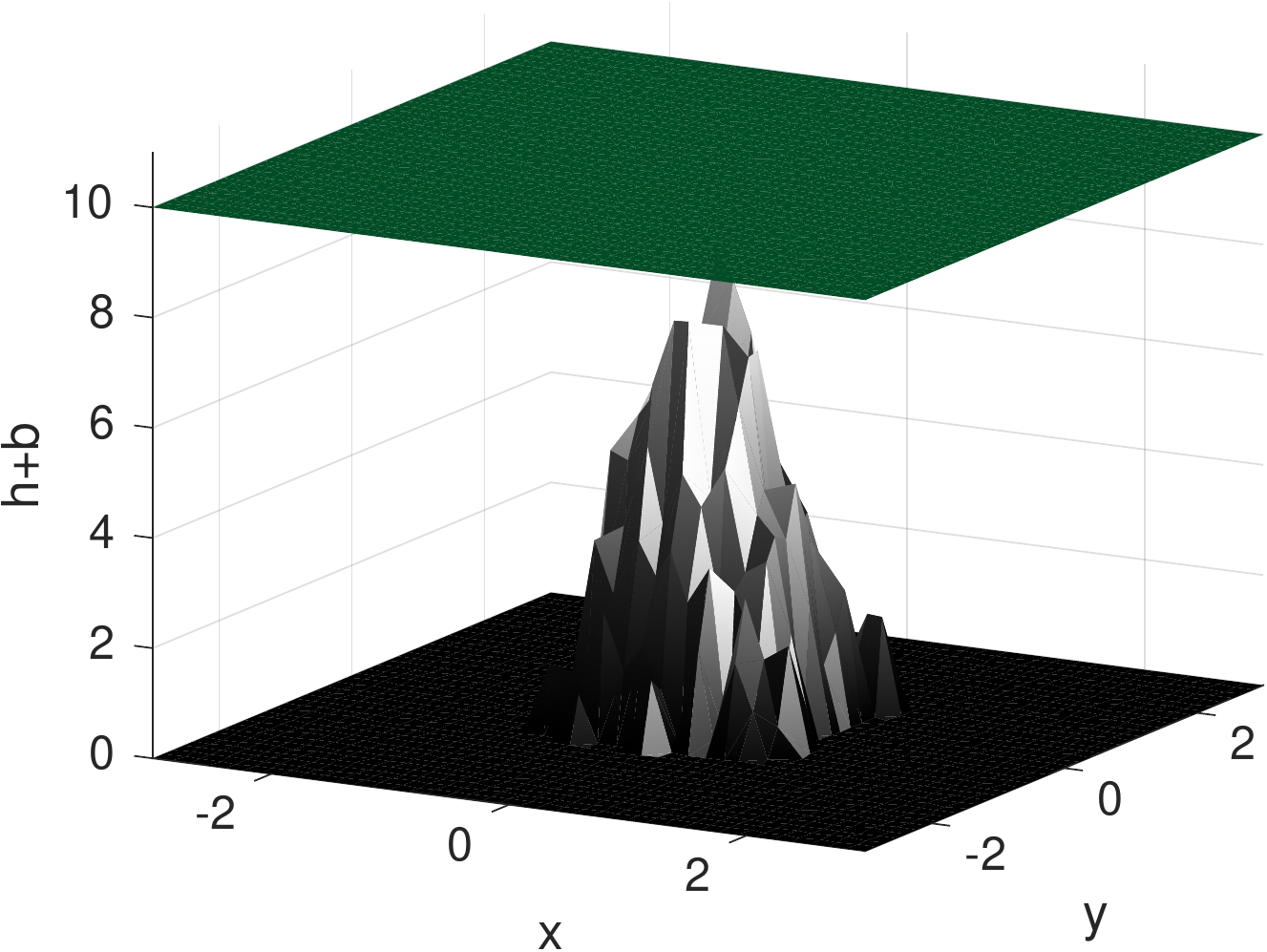}
    \end{subfigure}\hspace{1cm}
\begin{subfigure}[b]{0.45\textwidth}
            \centering
            \includegraphics[width=\textwidth]{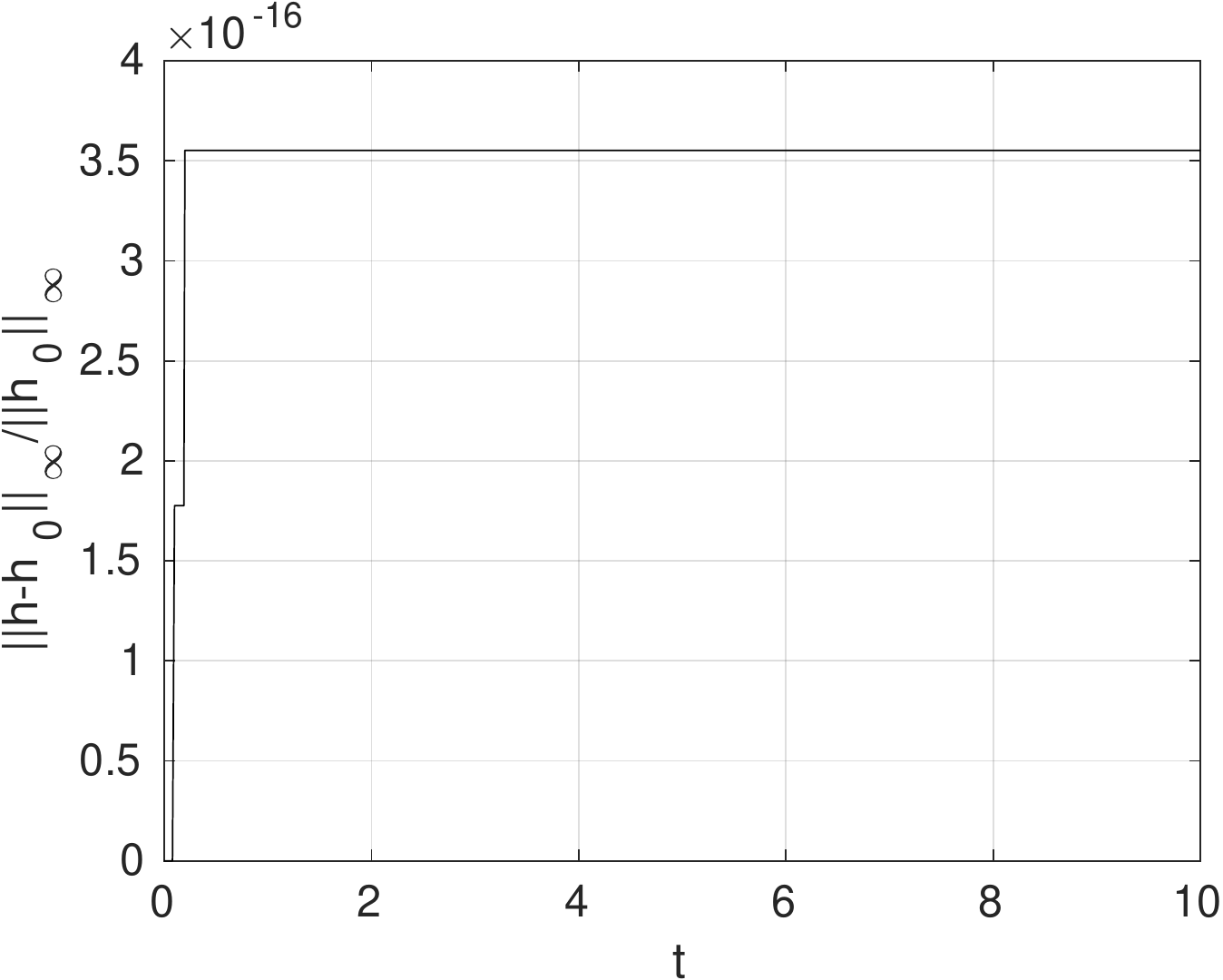}
    \end{subfigure}
    \caption{Numerical integration of the shallow-water equations using the RBF-FD method on $n=1600$ irregularly spaced nodes, integrated up to $t=10$ with the Heun scheme. \textbf{Left:} Total water height at $t=10$ and bottom topography. \textbf{Right:} Relative $l_\infty$-error in the water height.}
    \label{fig:RBF2D}
\end{figure}

Stabilization of the scheme was done by adding hyperviscosity of the
form $-(1)^{k+1}\nu\Delta^k (uh)$ and $-(1)^{k+1}\nu\Delta^k (vh)$ to
the momentum equations in the $x$- and $y$-directions,
respectively. We chose $k=2$ and experimentally set $\nu$ such that
the scheme remains stable but does not become unnecessarily
diffusive. Note that similar results were obtained when refining the
grid (i.e., starting from a finer uniform mesh but performing the same
random perturbations to both node location and bottom topography). In
this case, we replace the fixed hyperviscosity parameter, $\nu$, by $\nu(\Delta r)^{2k}$, where $\Delta r$ is a measure of the average nodal distance.  

\subsection{Parabolic bowl}
\red{
Having verified the numerical preservation of the lake at rest
solution, it is instructive to compare the numerical solution of the
balanced scheme to that of an unbalanced scheme for a more
challenging test case. Here, we consider oscillatory flow in a parabolic bowl, in the same setting as presented in~\cite{vate15a}, which is based on the exact solution derived in~\cite{thac81a}. In particular, we use the domain
$\Omega=[-5000,5000]$ with parabolic bottom topography $b=h_0(x/a)^2$,
where $a=3000$ and $h_0=10$. The initial conditions are such that the
exact solution to this benchmark test is given by
\begin{align*}
 &h_{\rm a}(t,x)=h_0-\frac{B^2}{4g}(1+\cos2\omega t)-\frac{Bx}{2a}\sqrt{\frac{8h_0}{g}}\cos(\omega t),\quad
 u_{\rm a}(t,x)=\frac{Ba\omega}{\sqrt{2h_0g}}\sin\omega t
\end{align*}
where $\omega=\sqrt{2gh_0}/a$ and $B=5$.
}

\red{
Using the three-point RBF-FD and Gaussian filter averaging rule
described above in Section \ref{ssec:1Dlake}, we integrate the
shallow-water equations numerically until $t=2000$.  We consider three
measures of the error for varying numbers of nodal points, $n$:
\begin{enumerate}
\item the maximum error in conservation of the total mass,
  $\mathcal{M}$, over all time steps,
\item the error in water height $h$, measured by taking the relative
  error in $h$ (measured in the maximum norm) at each time step, and
  measuring the maximum of these values over all time steps, and
  \item the error in momentum, $hu$, measured by taking the absolute
    error in $hu$ (measured in the maximum norm) at each time step,
    and measuring the maximum of these values over all time steps.
\end{enumerate}
Note that this solution requires an inundation model, since the water
surface hits the bowl and, thus, creates a moving boundary
condition. Inundation is not considered here, and we use the
analytical solution to prescribe the time-varying boundary
condition. For the discussion of a possible inundation model for this
case, consult~\cite{brec17a}. The results of this convergence study
are reported in Table~\ref{tab:ConvergenceStudyParabolicBowl}, showing
that, as expected, the balanced scheme is consistently better than the
unbalanced scheme, both being of second order.  The errors differ
most dramatically (by about one order of magnitude) for mass
conservation.
}

\begin{table}[!ht]
\renewcommand{\arraystretch}{1.4}
\centering
\caption{\footnotesize{Error in mass conservation $\mathcal M$, relative $l_\infty$-error for the total water height $h$, absolute $l_\infty$-error for the momentum $hu$ for the balanced and unbalanced schemes for the oscillatory flow in a parabolic bowl.}}
\begin{tabular}{|c||c|c||c|c||c|c|}
  \hline
  & \multicolumn{2}{|c||}{$\mathcal M$ error} &
  \multicolumn{2}{|c||}{$h$ error} & \multicolumn{2}{|c|}{$hu$
    error}\\
  \hline
$n$ & balanced & unbalanced & balanced & unbalanced & balanced & unbalanced\\
\hline
128 & $1.64\cdot 10^{-4}$ & $1.48\cdot 10^{-3}$ & $9.75\cdot10^{-4}$ & $2.44\cdot 10^{-3}$ & $8.43\cdot 10^{-2}$ & $1.91\cdot 10^{-1}$ \\
256 & $7.06\cdot 10^{-5}$ & $3.37\cdot 10^{-4}$ & $2.68\cdot 10^{-4}$ & $8.74\cdot 10^{-4}$ & $2.20\cdot 10^{-2}$ & $8.74\cdot 10^{-2}$ \\ 
512 & $8.21\cdot 10^{-6}$ & $7.90\cdot 10^{-5}$ & $7.20\cdot 10^{-5}$ & $1.77\cdot 10^{-4}$ & $5.38\cdot 10^{-3}$ & $1.31\cdot 10^{-2}$ \\
1024 & $2.43\cdot 10^{-6}$ & $1.77\cdot 10^{-5}$ & $1.88\cdot 10^{-5}$ & $6.26\cdot 10^{-5}$ & $1.36\cdot 10^{-3}$ & $3.3\cdot 10^{-3}$\\
\hline
\end{tabular}
\label{tab:ConvergenceStudyParabolicBowl}
\end{table}

\red{ For a second experiment, we consider fourth-order schemes and
  extend the comparison to include both a standard FD scheme and an
  RBF-FD scheme as described above.  Since we continue to use the Heun
  scheme for time stepping, we now decrease the time step by a factor
  of four each time the number of points is space is doubled, in order
  to balance the errors between the second-order time stepper and the
  fourth-order spatial discretizations.  We use uniform grids in
  space, and integrate to $t=1000$; for $n=64$ points in space, we use
  250 points in the time direction, which approximately balances the
  spatial and temporal discretization errors at this spatial mesh
  size.  For the standard FD scheme, we use the fourth-order
  (five-point) central difference stencil for the source derivative
  terms, the identity operator for the averaging rule, and the
  well-balanced prescription in
  \eqref{eq:ConditionsOnWellBalancedDerivativesMatrix} for the flux
  derivative.  For the RBF-FD scheme, we also use a five-point
  discretization for the source derivative, following the description
  in Section \ref{subsec:DerivativeApproximations}, but now including
  polynomials up to third order in the construction of the
  differencing scheme.  To achieve a nontrivial fourth-order averaging
  rule, we use a five-point averaging, with weights of $1.6$ for the
  point itself, $-0.4$ for the two immediate neighbours, and $0.1$ for
  the two distance-two neighbours.  We note that the normalized
  Gaussian filter used above cannot yield a fourth-order averaging, as
  some negative weights must appear in the averaging rule to attain
  fourth order.  We consider this scheme in both balanced form,
  following \eqref{eq:ConditionsOnWellBalancedDerivativesMatrix} to
  prescribe the flux derivative, and unbalanced form, directly using
  the fourth-order RBF-FD derivative for the flux term.  Numerical
  results are given in Table \ref{tab:ParabolicBowlFourthOrder}.  We
  note slight differences in the errors from the FD and balanced
  RBF-FD schemes, but these are small overall.  In comparison with the
  unbalanced RBF-FD scheme, however, we see notably larger errors in
  height and momentum, by up to a factor of three over the balanced
  schemes.  Most notably, comparing results for $n=256$ and $n=512$,
  we see reductions in both height and momentum error by factors of
  almost 15 for the two balanced schemes (consistent with
  fourth-order discretization), but only by a factor of 7 or 8 for the
  unbalanced scheme.  Taken together with the results for the
  lake at rest solution, these results indicate the advantage of
  choosing a well-balanced scheme for the shallow-water equations over
  an unbalanced scheme.  }

\begin{table}[!ht]
\renewcommand{\arraystretch}{1.4}
\centering
\caption{\footnotesize{Error in mass conservation $\mathcal M$,
    relative $l_\infty$-error for the total water height $h$, absolute
    $l_\infty$-error for the momentum $hu$ for the balanced and
    unbalanced schemes for the oscillatory flow in a parabolic bowl
    using fourth-order schemes.}}
\begin{tabular}{|c|c||c|c|c|c|}
  \hline
  & & $n=64$ & $n=128$ & $n=256$ & $n=512$ \\
  \hline
  \multirow{ 3}{*}{$\mathcal{M}$ error} & FD &
  $7.16\cdot 10^{-4}$ & $2.18\cdot 10^{-4}$ & $8.82\cdot 10^{-5}$ &
  $3.22\cdot 10^{-6}$\\
  & RBF-FD bal. &
  $7.33\cdot 10^{-4}$ & $1.66\cdot 10^{-4}$ & $5.57\cdot 10^{-5}$ &
  $5.47 \cdot 10^{-6}$\\
  & RBF-FD unbal. &
  $7.70\cdot 10^{-4}$ & $1.65\cdot 10^{-4}$ & $5.57\cdot 10^{-5}$ &
  $5.48\cdot 10^{-6}$\\
  \hline
  \multirow{ 3}{*}{$h$ error} & FD &
  $2.52\cdot 10^{-4}$ & $1.73\cdot 10^{-5}$ & $1.19\cdot 10^{-6}$ &
  $7.97\cdot 10^{-8}$\\
  & RBF-FD bal. &
  $2.87\cdot 10^{-4}$ & $2.09\cdot 10^{-5}$ & $1.45\cdot 10^{-6}$ &
  $9.84\cdot 10^{-8}$\\
  & RBF-FD unbal. &
  $2.72\cdot 10^{-4}$ & $2.61\cdot 10^{-5}$ & $1.68\cdot 10^{-6}$ &
  $2.06\cdot 10^{-7}$\\
  \hline
  \multirow{ 3}{*}{$hu$ error} & FD &
  $1.69\cdot 10^{-2}$ & $1.11\cdot 10^{-3}$ & $7.24\cdot 10^{-5}$ &
  $4.92\cdot 10^{-6}$\\
  & RBF-FD bal. &
  $1.88\cdot 10^{-2}$ & $1.24\cdot 10^{-3}$ & $8.15\cdot 10^{-5}$ &
  $5.56\cdot 10^{-6}$\\
  & RBF-FD unbal. &
  $1.48\cdot 10^{-2}$ & $2.35\cdot 10^{-3}$ & $1.01\cdot 10^{-4}$ &
  $1.52\cdot 10^{-5}$\\
  \hline
\end{tabular}
\label{tab:ParabolicBowlFourthOrder}
\end{table}

\section{Conclusion}\label{sec:ConclusionsShallowWater}

In this paper, we have derived general criteria for obtaining well-balanced numerical schemes for the shallow-water equations that employ a nodal expansion for their spatial derivative approximation. We have shown analytically and numerically that the resulting discretization schemes for the shallow-water equations exactly maintain the lake at rest steady state, which is considered as the first important criterion for applying such schemes to real-world problems such as tsunami modeling. We have further proved consistency and order conditions for the discrete differential and averaging operators involved in these well-balanced schemes.

One particularly important feature of the derived schemes is that they do not require the nodes to lay on a uniform, orthogonal grid. Rather, any nodal distribution can be used and the resulting schemes will remain well-balanced. This feature is important as it guarantees that various adaptation strategies, such as $h$- and $r$-adaptivity \red{(i.e.\ introducing or removing nodes, as well as dynamically redistributing them)}, can be used without complicating the design of the resulting numerical method. Due to the involved time and length scales in the shallow-water equations when used for tsunami modeling (e.g.\ open ocean wave propagation vs.\ coastal inundation), adaptivity is usually a practical necessity~\cite{geor06By}.

The present work is also an important step in the development of mimetic methods for general meshless numerical schemes, in that we have derived criteria that derivative approximations have to mimic in order for the resulting numerical scheme to be well-balanced. More generally, mimetic discretization is an active field of research in which one aims to discretize differential equations is such a way that certain important identities from vector calculus will be preserved in a numerical scheme. This is important since these vector identities are typically associated with central conservation laws in the equations of hydrodynamics and electrodynamics. For an overview of mimetic discretization schemes and some examples, including the shallow-water equations, consult e.g.~\cite{boch06Ay,brez05a,hyma99a,vanr11a,vant12a,vers03Ay}. Most mimetic methods derived so far apply to the finite difference, finite element and finite volume methodologies only and thus exclude the important class of meshless integration schemes. As these schemes are getting increasingly popular in fields such as atmospheric sciences, ocean sciences and geophysics, whose governing equations all admit important conservation laws, finding mimetic schemes in the wider meshless framework is an important timely research field. \red{For first results of this research perspective in geophysics, see~\cite{mart15a,mart17a}. The present study for the shallow-water equations can thus be regarded as being amongst the first examples for mimetic methods within the meshless methodology.}

\section*{Acknowledgments}

This research was undertaken, in part, thanks to funding from the Canada Research Chairs program and the NSERC Discovery Grant program. \red{The authors thank Grady Wright for helpful discussions, and the two anonymous referees for their helpful and considerate remarks.}

\footnotesize\setlength{\itemsep}{0ex}

\end{document}